\documentclass[APA,LATO1COL]{WileyNJD-v2}

\usepackage{amsfonts}
\usepackage{lipsum}
\DeclareMathOperator*{\argmin}{\mathrm{argmin}}
\DeclareMathOperator*{\argmax}{\mathrm{argmax}}

\articletype{}%

\received{}
\revised{}
\accepted{}

\newcommand{\R}{\mathbb{R}}
\newcommand{\plsEstim}{\hat{\boldsymbol\beta}_{\mathrm{PLS}}^{\scriptscriptstyle {(L)}}}
\newcommand{\olsEstim}{\hat{\boldsymbol \beta}_{\mathrm{OLS}}}
\newcommand{\xcov}{{\boldsymbol \Sigma_{\text{\tiny XX}}}}
\newcommand{\xycov}{{\boldsymbol \Sigma_{\text{\tiny XY}}}}
\newcommand{\matW}{\mathbf W_L} 
\newcommand{\matP}{\mathbf P_L}
\newcommand{\matD}{\mathbf D_L} 
\newcommand{\matT}{\mathbf T_L}
\newcommand{\matR}{\mathbf R_L}

\begin{document}

\def\mytitle{Relation between PLS and OLS regression in terms of the eigenvalue distribution
  of the regressor covariance matrix }

\title{\mytitle}

\author[1]{David del Val}

\author[2]{José R. Berrendero}

\author[3]{Alberto Suárez}

\address[1]{Departamento de Matemáticas and Departamento de Ingeniería Informática, 
  Universidad Autónoma de Madrid}
\address[2]{Departamento de Matemáticas, Universidad Autónoma de Madrid,
  and Instituto de Ciencias Matemáticas ICMAT (CSIC-UAM-UCM-UC3M)}
\address[3]{Departamento de Ingeniería Informática, Escuela Politécnica Superior, %
  Universidad Autónoma de Madrid}

\corres{%
  \email{delvaldavid1@gmail.com} (D. del Val)
  \email{joser.berrendero@uam.es} (J.R. Berrendero) \linebreak
  \email{alberto.suarez@uam.es} (A. Suárez)
}

\abstract[Summary]{%
  Partial least squares (PLS) is a dimensionality reduction technique introduced
  in the field of chemometrics and successfully employed in many other areas.
  The PLS components are obtained by maximizing the covariance between linear combinations of the regressors and of the target variables.
  In this work, we focus on its application to scalar regression problems. 
  PLS regression consists in finding the least squares predictor that is a linear combination of a subset of the PLS components.
  Alternatively, PLS regression can be formulated as a least squares problem restricted to a Krylov subspace.  
  This equivalent formulation is employed to analyze the distance between $\plsEstim$, the PLS estimator of the vector of
  coefficients of the linear regression model based on $L$ PLS components, and
  $\olsEstim$, the one obtained by ordinary least squares (OLS), as a function of $L$.
  Specifically, $\plsEstim$ is the
  vector of coefficients in the aforementioned Krylov subspace that is closest to
  $\olsEstim$ in terms of the Mahalanobis distance with respect to the covariance matrix of
  the OLS estimate. 
  We provide a bound on this distance that depends only
  on the distribution of the eigenvalues of the regressor covariance matrix.
  Numerical examples on synthetic and real-world data are used
  to illustrate how the distance between $\plsEstim$ and $\olsEstim$ depends on
  the number of clusters in which the eigenvalues of the regressor covariance
  matrix are grouped.
}


\jnlcitation{\cname{%
    \author{del Val D.}
    \author{Berrendero JR.}
    \author{Suárez A.}
  }(\cyear{2023}),
  \ctitle{%
    \mytitle
  }.}

\maketitle

\section{Introduction}\label{sec1}

  Partial least squares (PLS) is a family of dimensionality reduction methods
  introduced in the field of chemometrics \citep{woldNIPALS1977}, where it is
  extensively used \citep{frankChemometrics1993, woldPLSregression2001}. Its
  success in this discipline has led to its adoption in other scientific areas
  such as medicine \citep{worsleyOverview1997,nguyenTumorClassification2002},
  physiology \citep{lobaughSpatiotemporal2001}, and pharmacology
  \citep{nilssonMultiway1997}.
  In PLS, two blocks of random variables $X$ and $Y$ are considered. 
  The PLS components are built in a stepwise manner by maximizing the covariance between linear combinations of the components of $X$ and of $Y$ and imposing orthogonality to the previously identified components. 
  In variants of PLS used only for dimensionality reduction, $X$ and $Y$ are handled in a symmetric manner. 
  As a result, the maximum number of components is limited by
  the block that has the lowest dimension. 
  When PLS is employed for regression, non-symmetric variants are used to account for the distinct roles
  played by the two blocks: predictor ($X$) and response ($Y$) variables. 
  In contrast to the symmetric case, the number of components is limited only by the dimensionality of $X$, a feature that is of particular relevance in scalar regression problems.

  In this paper, we apply PLS to a random vector $X$ with $D$ components and
  a scalar random variable $Y$. 
  In this setting, PLS extracts a sequence of $L\le D$
  (random) components $\{t_l\}_{l=1}^L$, each of which is a linear combination of
  the coordinates of $X$. 
  The PLS components are orthogonal directions along
  which the covariance with the response variable, $Y$, is maximal. This
  optimization criteria is closely related to the ones used in other linear dimensionality
  reduction techniques such as principal component analysis (PCA) and canonical
  correlation analysis (CCA). 
  The CCA components are identified by maximizing the
  correlation with the response variable, instead of the covariance. In PCA, the
  principal components are defined solely in terms of the regressor variables:
  they are linear combinations of the coordinates of X, obtained sequentially by
  maximizing the variance in the space orthogonal to the one spanned by the
  previously identified components. The covariance of $t_l$, the $l$-th PLS
  component, and the response variable can be expressed in terms of the
  correlation and the corresponding variances:
  \begin{equation*}
    \label{eq:covariance}
    \mathrm{cov}(t_l, Y)^2 = \mathrm{var}(t_l) \mathrm{corr}(t_l,Y)^2 \mathrm{var}(Y).
  \end{equation*}
  This expression makes it clear that the PLS objective $(\mathrm{cov}(t_l,Y))$
  combines the optimization
  objective of PCA ($\mathrm{var}(t_l)$) and CCA ($\mathrm{corr}(t_l,Y)$).

  PLS was originally introduced in \cite{woldNIPALS1977}.
  In that work, the PLS components are defined computationally as the result of applying the NIPALS
  (non-linear iterative partial least squares) algorithm. 
  For scalar $Y$, the components identified by NIPALS are the
  solution of a constrained optimization problem.
  This problem consists in finding orthogonal
  linear combinations of the coordinates of $X$ that maximize the covariance with the response variable   \citep{dejongSIMPLS1993}. 
  Also in the case of scalar response, PLS regression with
  $L$ components can be formulated as a least squares problem restricted to
  the Krylov subspace of order $L$ generated by $\mathrm{cov}(X, X)$, the
  covariance matrix of X, and $\mathrm{cov}(X, Y)$, the vector of cross-covariances
  between X and Y \citep{hellandStructurePLS1988}. 
  As shown in
  Section \ref{sec:convergence}, $\plsEstim$ is the estimator in the Krylov
  subspace of order $L$ that is closest to $\olsEstim$, the ordinary least
  squares (OLS) estimator, in terms of the Mahalanobis distance with the covariance matrix of the OLS estimator. 
  Moreover, the PLS estimator is
  equal to the OLS one when $L$ is the number of distinct eigenvalues of
  $\mathrm{cov}(X,X)$. Given this equivalency, it is also possible to show that
  the conjugate gradient method \citep{woldPLSInverse1984} and the Lanczos
  bidiagonzalization algorithm \citep{eldenPLSLanczos2004} can be used as an
  alternative to NIPALS for PLS regression. 
  Another contribution of this work is
  to utilize the properties of the conjugate gradient method
  \citep{hestensCG1952, wrightNumerical1999} to quantify the differences between the  PLS and OLS estimators of the vector of regression coefficients, by establishing an equivalency between PLS and a  polynomial fitting problem. 
  From this reformulation of the problem it is possible to derive an upper bound for the distance between these estimators that depends only on the spectrum of $\mathrm{cov}(X,X)$. 
  In light of this analysis, we explore the relation between these estimators in terms of the characteristics of the distribution of eigenvalues. 
  In particular, if these eigenvalues are grouped into $k$
  tight clusters, PLS with $k$ components provides a good approximation to OLS.

  Finally, we carry out an empirical comparison of PCR (principal components
  regression) and PLS regression. This comparison shows that PLS and PCR are
  optimal for different types of eigenvalue distributions. 
  PLS performs best when
  the eigenvalues are clustered around a few values. 
  In contrast, PCA works best in the presence of a few dominant eigenvalues.

  The article is organized as follows: 
  In Section \ref{sec:dim_red}, PLS is introduced as a dimensionality reduction method.
  The NIPALS algorithm for PLS regression is also detailed, and its properties are analyzed. 
  The use of PLS in regression is described in Section
  \ref{sec:pls}. 
  Specifically, we provide a novel derivation of the equivalence
  between the standard formulation, as a least squares regression problem in the space spanned by the
  first PLS components, and one based on solving a least squares problem
  restricted to a Krylov space. 
  In Section \ref{sec:convergence}, the differences
  between the PLS estimator of the vector of regressor coefficients and the OLS one are quantified in terms of the Mahalanobis distance with the covariance matrix of the OLS estimator. 
  Proofs of some of these relations are given in the Appendix.
  In Section \ref{sec:examples}, the results of a numerical investigation of the performance of PLS are presented for different scenarios using synthetic and real-world regression problems. 
  Finally, the conclusions of this work are presented in Section \ref{sec:conclusions}.  

\section{Dimensionality reduction with partial least squares}
  \label{sec:dim_red}

  Consider the sample $\{(\mathbf x_i, y_i)\}_{i=1}^N$, whose $i$-th element is
  characterized by $\mathbf x\in \mathbb R^D$ and $\mathbf y_i\in \mathbb R$. The
  goal of PLS is to identify a set of components $\{\mathbf t_l\}_{l=1}^L$ that
  capture linear relations between X and Y. Originally, the PLS components were
  introduced as the output of the NIPALS algorithm, which will be described later
  in this section. In \cite{dejongSIMPLS1993}, it is shown that the $l$-th PLS
  component is the solution of the following optimization problem:
  \begin{equation}
    \label{eq:pls:optimization}
    \begin{alignedat}{4}
      \mathbf t_l =
       &  & \ \argmax_{\mathbf t} \quad \mathrm{cov}(\mathbf t, \mathbf y)
      \quad
      \mathrm{subject\ to} \quad
       &  &                                                                &
      \mathbf t = \mathbf X \mathbf r,\ \  \mathbf r \in\mathbb{R}^D, \ \
      \|\mathbf r\|=1;
      \\
       &  &                                                                &   
       &  & \mathbf t^\top \mathbf t_i=0 \quad i=1,\dots,\, l-1,
    \end{alignedat}
  \end{equation}
  where $\mathbf X=(\mathbf x_1,\dots \mathbf x_N)^\top\in \R^{N\times D}$ and
  $\mathbf y=(y_1,\dots y_N)^\top$.
  Associated with the components, the weight vectors $\{\mathbf r_l\}_{l=1}^L$
  are defined so that $\|\mathbf r_l\|=1$ and $\mathbf t_l = \mathbf X \mathbf
    r_l$, for $l=1,\dots, L$.

  The PLS components can be computed using different algorithms. The NIPALS
  algorithm, which we will describe in the remainder of this section, was the
  first one introduced and is still widely used today. This algorithm was first
  introduced in \cite{woldNIPALS1977}, and has been the object of successive
  refinements \citep{wegelinSurvey2000}. In this work we focus on the NIPALS
  algorithm for scalar PLS regression \citep{eldenPLSLanczos2004,
    rosipalAdvances2005}. NIPALS follows an iterative approach. At the end of each
  iteration, $\mathbf X$, the data matrix, is modified by removing the projection
  on the component computed in that iteration (line 7). As a result, a sequence
  of projections of the data matrix can be considered: $\{\mathbf X_l\}_{l=1}^L$.
  This deflation step ensures that subsequent components computed by the
  algorithm are orthogonal to the ones extracted up to that point. In particular,
  from lines 6 and 7, $\mathbf X_l = \Big( \mathbf I- \frac{\mathbf t_l \mathbf
      t_l^\top}{\mathbf t_l^\top \mathbf t_l} \Big) \mathbf X_{l-1}$, for
  $l=1,\dots L$.

  The $l$-th component can be computed as $\mathbf t_l= \mathbf X\mathbf r_l$,
  the projection of the original data onto the direction defined by the vector of
  weights $\mathbf r_l$. Alternatively, it is $\mathbf t_l = \mathbf X_{l-1}
  \mathbf w_l$, where $\mathbf w_l$ is the $l$-th weight vector extracted
  by NIPALS. Finally, the regressor and response loadings are defined as
	$
    \mathbf p_l = \mathbf X_{l-1}^\top \mathbf t_l /\|\mathbf t_l\|^2
	$
	and
	$
    q_l= \mathbf y^\top \mathbf t_l / \|\mathbf t_l\|^2,
	$
	for
	$
    l=1,\dots, L,
	$
  respectively.
  \begin{algorithm}[h!]
    \vspace{2pt}
    \hspace{0.025\linewidth}%
    \begin{minipage}[t]{0.07\linewidth}
      \textbf{Input}
    \end{minipage}
    \begin{minipage}[t]{0.3\linewidth}
      $\mathbf X$: the regressor variable data matrix. \\
      $\mathbf y$: the response variable data vector. \\
      $L$: the number of components to extract.
    \end{minipage}
    \hfill
    \begin{minipage}[t]{0.07\linewidth}
      \textbf{Output}
    \end{minipage}
    \begin{minipage}[t]{0.25\linewidth}
      $\{\textbf w_l\}_{l=1}^L$: projection weights. \\
      $\{\textbf t_l\}_{l=1}^L$: components.  \\
      $\{\textbf p_l\}_{l=1}^L$: loadings.
    \end{minipage}
    \vspace{5pt}
    \\
    \hrule
    \begin{algorithmic}[1]
      \State $\mathbf X_0 \gets \mathbf X$
      \State $l\gets1$
      \While {$l<L$}
      \State $\mathbf w_l \gets \mathbf X_{l-1}^\top \mathbf y
        /\|\mathbf X_{l-1}^\top\mathbf y\|$
      \Comment{\textit{Weights} calculation}%

      \State $\mathbf t_l \gets \mathbf X_{l-1} \mathbf w_l$
      \Comment{\textit{Scores} calculation}%

      \State $\mathbf p_l \gets \mathbf X_{l-1}^\top \mathbf t_l
        / (\mathbf t_l^\top \mathbf t_l)$
      \Comment{\textit{Loadings} calculation}

      \State $\mathbf X_l \gets \mathbf X_{l-1} - \mathbf t_l \mathbf p_l^\top$
      \Comment{Deflate $\mathbf X$}
      \State $l \gets l+1$
      \EndWhile
    \end{algorithmic}
    \caption{NIPALS for PLS regression with scalar response} \label{alg:nipals:pls1}
  \end{algorithm}
  \vspace*{-10pt}

  From this algorithm, we can derive a series of properties. 
  The ones relevant
  for the rest of this paper are included in the following propositions, whose
  proofs are given in the appendix.

  \begin{proposition}
    \label{prop:nipals_properties}
    From the NIPALS algorithm, the following properties can be derived:
    \begin{enumerate}
      \item In terms of the PLS components, the original data can be expressed as
            \begin{equation}
              \label{eq:nipals:approx}
              \mathbf X = \mathbf T_L \mathbf P_L^\top + \mathbf X_L,
              \qquad
              \mathbf y = \mathbf T_L \mathbf Q_L^\top + \mathbf y_L,
            \end{equation}
            where $\mathbf X_L\in \mathbb R^{N\times D}$ and $\mathbf y_L \in \mathbb R^N$
            are defined as
            \begin{equation}
              \label{eq:nipals:x_def}
              \mathbf X_L = \prod_{i=1}^{L} \left(
              \mathbf I - \frac{\mathbf t_i \mathbf t_i^\top}
              {\mathbf t_i^\top \mathbf t_i}
              \right) \mathbf X,
              \qquad
              \mathbf y_L = \prod_{i=1}^{L} \left(
              \mathbf I - \frac{\mathbf t_i \mathbf t_i^\top}
              {\mathbf t_i^\top \mathbf t_i}
              \right) \mathbf y.
            \end{equation}
            Additionally, $\mathbf T_L$, $\mathbf P_L$ and $\mathbf Q_L$ are defined as
            $\mathbf T_L = (\mathbf t_1, \dots \mathbf t_L)\in \mathbb R^{N\times L}$,
            $\mathbf P_L = (\mathbf p_1,\dots, \mathbf p_L) \in \mathbb{R}^{D\times L}$,
            and $\mathbf Q_L = (q_1,\dots, q_L)\in \mathbb{R}^{1\times L}$.

      \item The Frobenius norms of $\mathbf X_L$ and $\mathbf y_L$ decrease as $L$
            increases.

      \item After $L$ iterations, $\mathbf X_L$ is orthogonal to the weights: $\mathbf X_L
              \mathbf W_L = \mathbf 0$, where $\mathbf W_L=(\mathbf w_1,\dots \mathbf w_L)\in
              \mathbb R^{M\times L}$

      \item The loading matrices $\mathbf P_L$ and $\mathbf Q_L$ can be expressed in terms
            of the components and the original data as $\mathbf P_L = \mathbf X ^\top
              \mathbf T_L \mathbf D_L^{-2}$ and $\mathbf Q_L~=~\mathbf y^\top \mathbf T_L
              \mathbf D_L^{-2}$, with 
              $\mathbf D_L = \text{diag} \left( \|\mathbf{t}_1 \|, \ldots,  \|\mathbf{t}_L\| \right) \in \mathbb{R}^{L\times L}$.
    \end{enumerate}
  \end{proposition}

  NIPALS calculates both the components and the weights needed to express the
  components as projections of the deflated $\mathbf X_L$ data matrices. However,
  it is advantageous to express the components as a projection of the original
  data $\mathbf X$ to simplify the resulting expressions. 
  The following proposition, whose proof can be found in the appendix, provides an expression for these projection directions
  \begin{proposition}
    \label{prop:rotation}
    The matrix $\mathbf R_L\in \mathbb R^{D\times L}$ that fulfills
    $\mathbf T_L = \mathbf X \mathbf R_L$ is
    $ \mathbf R_L = \mathbf W_L (\mathbf P_L^\top \mathbf W_L)^{-1} $.
  \end{proposition}

\section{Partial least squares regression with scalar response}
  \label{sec:pls}

  Consider the linear regression model $Y = \boldsymbol{\beta}^\top X\, +\,
    \epsilon$, where $X$ is the regressor vector with $D$ components,
  $\boldsymbol{\beta}$ is the vector of coefficients, which needs to be
  estimated, $\epsilon$ is random noise independent of $X$, and $Y$ is the scalar
  response. For the sake of simplicity, and without loss of generality, both $X$
  and $Y$ are assumed to have zero mean. To fit this model, $N$ independent
  observations drawn from this model are available: $\{(\mathbf{x}_i,
    y_i)\}_{i=1}^N$. We further assume that $\{\epsilon_i\}_{i=1}^N$ are iid
  with variance $\sigma^2$. In this setting, we seek to estimate a vector of
  coefficients $\boldsymbol \beta$ such that $ y_i = \boldsymbol\beta^\top
    \mathbf x_i + \epsilon_i$, where $i=1,\dots, N$. These equations can be grouped
  row-wise into the matrix equation
  \begin{equation}
    \label{eq:sample_reg}
    \mathbf y = \mathbf X \boldsymbol \beta + \boldsymbol \epsilon,
  \end{equation}
  where $\mathbf y=(y_1,\dots, y_N)^\top\in \mathbb R^{N}$,
  $\mathbf X = (\mathbf x_1, \dots, \mathbf x_N)^\top \in \mathbb R^{N\times D}$
  and $\boldsymbol \epsilon = (\epsilon_1,\dots, \epsilon_N)^\top \in \mathbb R^N$.

  One possible estimator for ${\boldsymbol \beta}$ is the ordinary least squares
  estimator (OLS), given by
  \begin{equation}
    \label{eq:ols_def}
    \hat{\boldsymbol \beta}_{\mathrm{OLS}} =
    \argmin_{\boldsymbol \beta\in \mathbb R^D}
    \|\mathbf y - \mathbf X \boldsymbol \beta \|^2 =
    (\mathbf X^\top \mathbf X)^{-1} \mathbf X^\top \mathbf y =
    \xcov^{-1} \xycov,
  \end{equation}
  where $\|\cdot\|$ is the euclidean norm,
  $\xcov= \mathbf X^\top \mathbf X$ is the empirical estimate of the covariance matrix
  of $X$ scaled by the number of observations,
  and
  $\xycov = \mathbf X^\top \mathbf y$
  is the empirical estimate of the covariance matrix
  of $X$ and $Y$ scaled by the number of observations as well.
 
  A different estimator of $\boldsymbol \beta$ is obtained using PLS regression.
  The first step is to extract $L$ PLS components as described in the
  previous section. 
  Then, a linear prediction is made in terms of these components:  
  $
      \sum_{l=1}^L \hat{\gamma_l}^{(L)} \mathbf t_l = \mathbf{T}_L \hat{\boldsymbol{\gamma}}^{(L)},
  $
  with 
  $\hat{\boldsymbol{\gamma}}^{(L)} = \left( \hat{\gamma}_1^{(L)}, \ldots, \hat{\gamma}_L^{(L)} \right)^{\top}$ 
  determined by least squares as
  \begin{equation}
  \hat{\boldsymbol \gamma}^{(L)}
  =  \argmin_{\boldsymbol \gamma \in \mathbb R^L} 
      \| \mathbf y - \mathbf T_L  \boldsymbol{\gamma} \|^2 
  =  (\mathbf T_L^\top \mathbf T_L)^{-1}\mathbf T_L^\top \mathbf y
  =  \mathbf D_L^{-2}\mathbf T_L^\top\mathbf y.
  \end{equation}
  The PLS estimator of $\boldsymbol\beta$ is obtained by expressing this linear predictor in terms of the original variables 
  $ \mathbf T_L \hat{\boldsymbol{\gamma}}^{(L)} 
    = \mathbf X \hat{\boldsymbol \beta}_{\mathrm{PLS}}^{(L)} 
  $.
  Using the definition of $\mathbf{R}_L$ from Proposition \ref{prop:rotation}:  
  $  
  \mathbf{X} \hat{\boldsymbol \beta}_{\mathrm{PLS}}^{(L)} = \mathbf{T}_L \hat{\boldsymbol{\gamma}}^{(L)} 
  =  \mathbf{X} \mathbf{R_L} \hat{\boldsymbol{\gamma}}^{(L)}.
  $
  Therefore, the PLS estimator of the vector of regression coefficients is
  \begin{equation}
    \label{eq:pls1_beta}
    \plsEstim = \mathbf R_L \hat{\boldsymbol \gamma}^{(L)}
    =
    \mathbf R_L \mathbf D_L^{-2}\mathbf T_L^\top \mathbf y.
  \end{equation}
  This estimator is then used to yield the linear prediction $\mathbf{X} \plsEstim$, in terms of the original variables.

  Alternatively, $\plsEstim$ can be viewed as the least squares estimator of $\boldsymbol\beta$ when the optimization is constrained to a Krylov subspace. 
  Krylov subspaces are defined as follows:
  \begin{definition}
    \label{def:kry}
    The Krylov subspace of order
    $L \le D$ generated by the matrix  
    $\mathbf A\in \mathbb{R}^{D\times D}$ and the vector
    $\mathbf b\in \mathbb{R}^D$, $\mathbf b\ne 0$ is
    \begin{equation}
      \label{eq:kry:def}
      \mathcal K_L(\mathbf A, \mathbf b)= \mathrm{span}\{
      \mathbf b,
      \mathbf A\mathbf b,
      \dots,
      \mathbf A^{L-1}\mathbf b
      \}.
    \end{equation}
  \end{definition}

  \begin{theorem}
    \label{thm:equiv}
    The PLS estimator with $L$ components defined in \eqref{eq:pls1_beta} is the
    solution to the least squares problem
    \begin{equation}
      \label{eq:pls_ols}
      \plsEstim
      =
      \argmin_{\boldsymbol \beta\in \mathcal K_L(\xcov ,\xycov)}
      \|\mathbf y - \mathbf X\boldsymbol\beta\|^2,
    \end{equation}
    where $\mathcal K_L(\xcov ,\xycov)$ is the Krylov
    subspace of order $L$ generated by the matrix $\xcov$ and the vector $\xycov$.
  \end{theorem}
  \begin{proof}
    Assume that the columns of $\mathbf B_L\in \R^{D\times L}$ constitute
    a basis of the Krylov subspace $\mathcal K_L(\xcov,\xycov)$.
    Then any $\boldsymbol\beta\in \mathcal K_L(\xcov,\xycov)$ can be expressed as
    $\boldsymbol\beta = \mathbf B_L \boldsymbol \alpha$ for some
    $\boldsymbol \alpha\in \mathbb R^L$. Thus, the constrained
    optimization problem given by \eqref{eq:pls_ols} can be transformed into
    an unconstrained optimization problem in $\mathbb R^L$:
    \begin{equation}
      \label{eq:kry:opt_formula}
      \argmin_{\boldsymbol{\beta}\in \mathcal K_L(\xcov, \xycov)}
      \|\mathbf y - \mathbf X \boldsymbol \beta \| ^2
      =
      \argmin_{\boldsymbol{\alpha}\in \R^L}
      \|\mathbf y - \mathbf X \mathbf B_L \boldsymbol \alpha \| ^2
      = \mathbf B_L (\mathbf B_L^\top \mathbf X^\top \mathbf X \mathbf B_L)^{-1}
      \mathbf B_L^\top \mathbf X^\top
      \mathbf y.
    \end{equation}

    As shown in \cite{eldenPLSLanczos2004}, the columns of the matrix
    $\mathbf W_L$ obtained after $L$ iterations of NIPALS constitute a basis of
    $\mathcal K_L(\xcov, \xycov)$.
    Therefore,  \eqref{eq:kry:opt_formula} holds for $\mathbf B_L = \mathbf W_L$.
    It is then possible to show that $\plsEstim$ can be expressed in the form given by the rhs of  \eqref{eq:kry:opt_formula} 
    with $\mathbf B_L = \mathbf W_L$.
    To this end, Propositions \ref{prop:nipals_properties} and \ref{prop:rotation} are applied repeatedly to \eqref{eq:pls1_beta}:
    
    \begin{equation*}
      \begin{alignedat}{1}
        \plsEstim
         & =
        \matR \matD^{-2}\matT^\top \mathbf y
        =
        \matW(\matP^\top \matW^{-1})\matD^{-2}\matR^\top \mathbf X^\top \mathbf y
        =
        \matW(\matP^\top \matW)^{-1}\matD^{-2}(\matW (\matP^\top \matW)^{-1})^\top\mathbf X^\top \mathbf y
        =    \\     & =
              \matW(\matW^\top \matP \matD^{2}\matP^\top \matW)^{-1}\matW^\top\mathbf X^\top \mathbf y
              =
              \matW(\matW^\top \mathbf X^\top \matT\matP^\top \matW)^{-1}
        \matW^\top\mathbf X^\top \mathbf y
		= \\ &=
              \matW(\matW^\top \mathbf X^\top (\mathbf X - \mathbf X_L) \matW)^{-1}
        \matW^\top\mathbf X^\top \mathbf y
        =
              \matW(\matW^\top \mathbf X^\top \mathbf X \matW)^{-1}
        \matW^\top\mathbf X^\top \mathbf y,
      \end{alignedat}
    \end{equation*}
    where the last step holds because of the orthogonality between $\mathbf X_L$ and
    $\mathbf W_L$ (Proposition \ref{prop:nipals_properties}).
  \end{proof}

  Other approaches can be adopted to prove this theorem. In
  \cite{eldenPLSLanczos2004}, the proof is based on the relation of PLS with the
  Lanczos bidiagonalization algorithm. An alternative derivation is given in
  \cite{takanePLSAlgorithm2016}, leveraging the properties of some bidiagonal and
  tridiagonal matrices in the NIPALS algorithm \citep{woldNIPALS1977}. However,
  in this proof, only simpler relationship between the matrices that are defined
  in the NIPALS algorithm are needed.

  The expression of the vector of PLS regression coefficients given by
  \eqref{eq:pls_ols} opens up the possibility of using numerical optimization
  algorithms that accept linear constraints to compute $\plsEstim$. It suffices
  to minimize $\|\mathbf y-\mathbf X \boldsymbol \beta \|^2$ subject to
  $\boldsymbol \beta$ belonging to $\mathcal K_L(\xcov,\xycov)$. In particular,
  the conjugate gradient algorithm is an iterative algorithm that minimizes a
  quadratic form $\psi(z)= \mathbf z^\top \mathbf A \mathbf z - \mathbf b^\top
    \mathbf z$ while exploring $\mathcal K_L(\mathbf A, \mathbf b)$ in the $L$-th
  iteration \citep{wrightNumerical1999}. 
  Thus, the optimization problem in Theorem \ref{thm:equiv} can be solved using  the conjugate gradient algorithm with $\mathbf A = \xcov$ and $\mathbf b=\xycov$. 
  In the next section, we take advantage of this observation to study how the PLS estimator approximates the OLS one.

  Additionally, the following theorem establishes an important link between the
  OLS and the PLS estimators.

  \begin{theorem}
    \label{thm:ols}
    The OLS estimator is contained in $\mathcal K_M(\xcov, \xycov)$, where $M$ is the number
    of distinct eigenvalues of $\xcov$.
  \end{theorem}
  \begin{proof}
    As a consequence of the Cayley-Hamilton theorem \citep[e.g.][p.220]{bronsonMatrixCalculus2009},
    since  $\xcov$ is a
    non-singular symmetric matrix, there exists a polynomial $P_{\xcov}$
    of degree $M-1$ such that $ P_{\xcov}(\xcov ) \xcov = \mathbf I$,
    where $M$ is the number of different eigenvalues of $\xcov$. Applying
    this result to the usual formula of OLS, we obtain
    $\hat{\boldsymbol \beta}_{\mathrm{OLS}}
      =
      (\mathbf X^\top \mathbf X)^{-1}\mathbf X^\top \mathbf y
      =
      P_{\xcov}(\xcov)\xycov
      \in \mathcal{K}_{M}(\xcov,\xycov)$.
    \qedhere
  \end{proof}
  \begin{corollary}
    The PLS estimator coincides with the OLS estimator
    after $M$ iterations, where $M$ is the number of
    different eigenvalues of $\xcov$:
    \begin{equation}
      \hat{\boldsymbol\beta}_{\mathrm{PLS}}^{(M)} =
      \hat{\boldsymbol\beta}_{\mathrm{OLS}}.
    \end{equation}
  \end{corollary}
  \begin{proof}
    It is a direct consequence of Theorem \ref{thm:ols} and the definition
    of $\plsEstim$ as a restricted least squared estimator in Theorem \ref{thm:equiv}.
  \end{proof}

\section{Relation between partial least squares and ordinary least squares}
  \label{sec:convergence}
  As described in the previous section, the PLS estimator of the vector of coefficients of a linear regression model with L components
  converges to the ordinarly least squares estimator as $L$ increases.
  Furthermore, they coincide when $L \ge M$, the number of distinct eigenvalues of
  $\xcov$.
  The goal of this section is to provide
  an upper bound for the distance between $\plsEstim$ and $\olsEstim$.
  In order to do so, we take advantage of the formulation of PLS in Theorem \ref{thm:equiv},
  as a constrained optimization problem that can be solved using conjugate gradients.
  The first part of this section follows the convergence analysis for the conjugate
  gradient method in \cite{wrightNumerical1999}.
  First, the PLS estimator is defined as the solution of  yet another optimization problem
  in which a distance to the OLS estimator is minimized subject to some constrains.
  \begin{proposition}
    \label{prop:pls_optimization_ols}
    The PLS estimator of the vector of coefficients of a linear regression model with $L$ components is the solution to the optimization problem
    \begin{equation}
      \label{eq:estimator_norm}
      \plsEstim =
      \argmin_{\boldsymbol \beta\in \mathcal K_{L}(\xcov, \xycov)}
      \Big\|
      \boldsymbol \beta - \hat{\boldsymbol \beta}_{\mathrm{OLS}}
      \Big\|_\xcov^2,
    \end{equation}
    where $\|\mathbf{z} \|^2_{\xcov} = \mathbf z^\top \mathbf \xcov \mathbf z$, 
    the square of the quadratic-form norm with the positive definite
    matrix $\xcov$.
  \end{proposition}
  \begin{proof}
    This result is a consequence of the definition of the PLS estimator with
    $L$ components provided in Theorem \ref{thm:equiv}:
    \begin{equation*}
      \label{eq:a_norm_phi_proof}
      \begin{alignedat}{1}
        \plsEstim
         & =
        \argmin_{\boldsymbol \beta\in \mathcal K_{L}(\xcov, \xycov)}
        \|\mathbf y -\mathbf X \boldsymbol \beta \|^2
        =
        \argmin_{\boldsymbol \beta\in \mathcal K_{L}(\xcov, \xycov)}
        \left(
        \mathbf y^\top \mathbf y - 2 \boldsymbol\beta^\top \mathbf X^\top \mathbf y
        + \boldsymbol \beta ^\top \mathbf X^\top \mathbf X \boldsymbol \beta
        \right)
        =
        \argmin_{\boldsymbol \beta\in \mathcal K_{L}(\xcov, \xycov)}
        \left(
        \boldsymbol \beta ^\top \mathbf X^\top \mathbf X \boldsymbol \beta
        - 2 \boldsymbol\beta^\top \mathbf X^\top \mathbf y
        \right)
        =    \\
         & =
        \argmin_{\boldsymbol \beta\in \mathcal K_{L}(\xcov, \xycov)}
        \left(
        \boldsymbol \beta ^\top \mathbf \xcov\,\boldsymbol \beta
        - 2 \boldsymbol\beta^\top  \xcov\, \olsEstim
        + \olsEstim^\top \xcov\, \olsEstim
        \right)
        =
        \argmin_{\boldsymbol \beta\in \mathcal K_{L}(\xcov, \xycov)}
        \|\boldsymbol \beta - \olsEstim \|_\xcov^2,
      \end{alignedat}
    \end{equation*}
    where we have used that $\mathbf X^\top \mathbf y = \mathbf X^\top \mathbf X \olsEstim = \xcov \olsEstim$.
  \end{proof}
  The quadratic-form norm $\| \cdot \|_{\xcov}$ is related to  the Mahalanobis distance with the covariance matrix of the OLS estimator of $\boldsymbol{\beta}$
  The following observation motivates the use of this norm as a natural way to quantify the differences between 
  $\plsEstim$ and $\olsEstim$.
  \begin{corollary}
    The PLS estimator of the vector of coefficients of a linear regression model with $L$ components is the solution of the optimization problem
    \begin{equation}
      \label{eq:estimator_norm_mah}
      \plsEstim =
      \argmin_{\boldsymbol \beta\in \mathcal K_{L}(\xcov, \xycov)}
      d_M(\boldsymbol \beta, \hat{\boldsymbol\beta}_{\mathrm{OLS}}),
    \end{equation}
    where $d_M$ is the Mahalanobis distance with respect to the matrix $\frac{1}{\sigma^2} \xcov^{-1}$, which is the covariance matrix
    of the OLS estimator of the regression coefficients conditioned to the observations of $X$.
  \end{corollary}
  \begin{proof}
    From \eqref{eq:ols_def},
    the variance of the OLS estimator conditioned to $\mathbf x_1, \dots \mathbf x_N$ is
	 $
      \mathbf{C}_{\mathrm{OLS}}=
      \mathrm{var}(\hat{\boldsymbol \beta}_{\mathrm{OLS}}| \mathbf x_1, \dots \mathbf x_N)=
      \sigma^2 (\mathbf X^\top \mathbf X)^{-1},
	 $
    where we have used that the $\mathrm{var}(\mathbf y | \mathbf x_1, \dots \mathbf x_N)=
      \mathrm {var}(\boldsymbol \epsilon)$, and that the observations of $\epsilon$ are iid random variables with variance $\sigma^2$.
    As a result, the squared Mahalanobis distance between the
    $\hat{\boldsymbol \beta}_{\mathrm {OLS}}$ estimator and
    some other estimator $\hat{\boldsymbol \beta}$ can be expressed as
    \begin{equation*}
      d_M(\hat{\boldsymbol \beta}, \hat{\boldsymbol \beta}_{\mathrm{OLS}})^2
      =
      (\hat{\boldsymbol \beta}- \hat{\boldsymbol \beta}_{\mathrm{OLS}})^\top
      \mathbf C^{-1}_{\mathrm{OLS}}
      (\hat{\boldsymbol \beta} - \hat{\boldsymbol \beta}_{\mathrm{OLS}})
      =
      \frac{1}{\sigma^2}
      (\hat{\boldsymbol \beta}- \hat{\boldsymbol \beta}_{\mathrm{OLS}})^\top
      (\mathbf X^\top \mathbf X)
      (\hat{\boldsymbol \beta} - \hat{\boldsymbol \beta}_{\mathrm{OLS}})
      = \frac{1}{\sigma^2}
      \|
      \hat{\boldsymbol \beta} - \hat{\boldsymbol \beta}_{\mathrm{OLS}}
      \|_\xcov^2.
    \end{equation*}
    Thus, the distance induced by the quadratic form norm $\|\cdot \|_{\xcov}$
    is proportional to the
    Mahalanobis distance with $\sigma^2 \xcov^{-1}$, the covariance matrix of the OLS estimator.
  \end{proof}

  Therefore, with $L$ components, PLS finds the closest estimator to $\olsEstim$
  with respect to the Mahalanobis distance with the covariance matrix of the OLS estimator in the Krylov subspace characterized by $\xcov$ and $\xycov$ of order $L$. 
  The Mahalanobis distance provides a natural measure of differences in the space of estimators, one that captures its geometry better than the Euclidean distance. For once, the Mahalanobis distance between the estimators is deeply related to the 
  euclidean distance between the predictions:
  \begin{equation*}
    d_M(\plsEstim, \olsEstim)^2 = \frac{1}{\sigma^2} 
    (\plsEstim - \olsEstim)^\top (\mathbf X^\top \mathbf X)
    (\plsEstim - \olsEstim)=
    \frac{1}{\sigma^2}
    \left\|\hat{\mathbf y}_{\mathrm{OLS}} - 
    \hat{\mathbf y}_{\mathrm{PLS}} 
    \right\|^2.
  \end{equation*}
  
  Additionally, the structure of Krylov subspaces  makes it possible to identify each element in a Krylov subspace of 
  order $L$ with a polynomial of order $L-1$. 
  As a result, the optimization problem in
  Proposition \ref{prop:pls_optimization_ols}, is equivalent to the optimization
  problem given in the following Corollary:

  \begin{corollary}
    \label{cor:polynomial_equiv}
    The PLS estimator with $L$ component is
    $\plsEstim = P_{L-1}^* (\xcov) \xycov$, where
    \begin{equation}
      \label{eq:polynomial_problem}
      P_{L}^* = \argmin_{P\in\mathcal P_L}
      \Big\|
      P(\xcov) \xycov -
      \hat{\boldsymbol \beta}_{\mathrm{OLS}}
      \Big\|_\xcov^2,
    \end{equation}
    and $\mathcal P_L$ is the space of polynomials of degree lower
    or equal to $L$.
  \end{corollary}
  \begin{proof}
    Since $\plsEstim\in\mathcal K_L(\xcov, \xycov)$,
    it can be expressed as
    $\plsEstim = P_{L-1}^* (\xcov)\xycov$. By substituting this expression
    into \eqref{eq:estimator_norm}, we obtain \eqref{eq:polynomial_problem}.
  \end{proof}

  As stated in the following theorem, whose proof is given in the appendix, the difference between the PLS estimator and the OLS estimator can be expressed as an optimization problem in a space of polynomials:
  \begin{theorem}
    \label{thm:error_polynomial}
    The distance between the PLS estimator and the OLS estimator fulfills
    \begin{equation}
      \label{eq:a_dist_optimization}
      \Big\|\plsEstim- \hat{\boldsymbol \beta}_{\mathrm{OLS}}\Big\|_\xcov^2
      =
      \min_{Q_L\in \Omega_L}
      \sum_{d=1}^D Q_L(\lambda_d)^2 \lambda_d \xi_d^2,
    \end{equation}
    where  $\{\lambda_d \}_{d=1}^D$ are the eigenvalues of $\xcov$, $\{\xi_d\}_{d=1}^D$ are the coefficients of
    the expansion of $\hat{\boldsymbol \beta}_{\mathrm{OLS}} $ in $\{\mathbf u_d\}_{d=1}^D$, the basis of eigenvectors of $\xcov$, and
    $\Omega_L=\{Q_L\in \mathcal P_L: Q_L(0) = -1\}$.
    Additionally, for each $L$, the minimum is reached for $Q_L^*(t) = t P_L^*(t) - 1$.
  \end{theorem}

  This theorem implies that any polynomial $Q_L \in \Omega_L$ can be used to provide an upper bound for the distance between the OLS and PLS estimators: 
    \begin{equation}
      \Big\|\plsEstim - \hat{\boldsymbol \beta}_{\mathrm{OLS}}\Big\|_\xcov^2
      \le
        \sum_{d=1}^D Q_L(\lambda_d)^2 \lambda_d \xi_d^2, \quad \text{for all } Q_L \in \Omega_L.
    \end{equation}
  Furthermore, it is possible to obtain an upper bound also in
  terms of the norm of the OLS estimator and of $Q_L$ evaluated at the eigenvalues of $\xcov$.

  \begin{corollary}
    Given a function
    $H: \Omega_L \rightarrow \mathbb{R}$
    that, for any polynomial $R\in \Omega_L$, fulfills $R(\lambda_d)^2\le H(R)$ over all $d=1,\dots,D$,
    and given a particular polynomial $Q_L\in \Omega_L$, 
    \begin{equation}
      \label{eq:bound_obtention}
      \Big\|\plsEstim - \hat{\boldsymbol \beta}_{\mathrm{OLS}}\Big\|_\xcov^2
      \le
      H(Q_L)
      \Big\| \hat{\boldsymbol \beta}_{\mathrm{OLS}} \Big\|^2_\xcov, \quad \text{for all } Q_L \in \Omega_L.
    \end{equation}
  \end{corollary}
  \begin{proof}
    From Theorem \ref{thm:error_polynomial}, and the condition $R(\lambda_d)^2 \le H(R) $ for $d = 1, \ldots, D$,
    \begin{equation*}
      \label{eq:bound_obtention_proof}
      \Big\|\plsEstim- \hat{\boldsymbol \beta}_{\mathrm{OLS}}\Big\|_\xcov^2
      =
      \min_{R\in \Omega_L}
      \sum_{d=1}^D R(\lambda_d)^2\lambda_d\xi_d^2
      \le
      \min_{R \in \Omega_L} H(R) \,    \sum_{d=1}^D \lambda_d\xi_d^2 
      = \min_{R\in \Omega_L} H(R) \, \Big\| \hat{\boldsymbol \beta}_{\mathrm{OLS}} \Big\|^2_\xcov
      \le
      H(Q_L)  \, \Big\| \hat{\boldsymbol \beta}_{\mathrm{OLS}} \Big\|^2_\xcov.
      \qedhere
    \end{equation*}

  \end{proof}

  Therefore, by choosing an $H$ function and a specific polynomial $Q_L$, an
  upper bound on the PLS error can be obtained. There are different choices for
  $H$. In \cite{wrightNumerical1999}, a number of results are given using the
  upper bound $H_1(Q_L) = \max_d Q_L(\lambda_d)^2$. However, this bound has a major
  disadvantage: it is not straightforward to calculate the polynomial $Q_L$ that
  minimizes $H_1$.  
  In the remainder of this section, the simpler upper bound $H_2(Q_L)=\sum_{d=1}^D Q_L(\lambda_d)^2$ is considered.
  The following theorem provides an uper bound on the PLS error by calculating the polynomial in $\Omega_L$ that minimizes $H_2$.
  \begin{theorem}
    \label{thm:bound}
    The following bound for the squared norm of the difference
    between
    the $L$-th PLS estimator and the OLS estimator holds:
    \begin{equation}
      \label{eq:final_bound}
      \Big\|\plsEstim - \hat{\boldsymbol\beta}_{\mathrm{OLS}}\Big\|_\xcov^2
      \le C_L
      \Big\|\hat{\boldsymbol\beta}_{\mathrm{OLS}}\Big\|_\xcov^2, \qquad
    \end{equation}
    where
    \begin{equation}
      \label{eq:bound_matrices}
      C_L = D(1-\mathbf c_L^\top \mathbf H_L^{-1} \mathbf c_L),
      \qquad
      \mathbf H_L = \begin{pmatrix}
        {\mu}'_2     & \dots  & {\mu}'_{L+1} \\
        \vdots       & \ddots & \vdots       \\
        {\mu}'_{L+1} & \dots  & {\mu}'_{2L}  \\
      \end{pmatrix},
      \qquad
      \mathbf c_L = \begin{pmatrix}
        {\mu}'_1 \\ \vdots \\ {\mu}'_L
      \end{pmatrix},
    \end{equation}
    and $\mu'_l$ is the $l$-th raw moment of
    the distribution of the eigenvalues of $\xcov$.
  \end{theorem}

  This result provides an upper bound for the distance between 
  $\plsEstim$ and $\hat{\boldsymbol\beta}_{\mathrm{OLS}}$
  that depends only on the distribution of the eigenvalues of the regressor covariance matrix.
  Explicit expressions of this bound can be derived for PLS regression with one and two components.
  \begin{corollary}
    \label{cor:explicit_bounds}
    The bounds given in \eqref{eq:final_bound} for $L=1$ and $L=2$ can be expressed as a function of
    the coefficient of variation ($c_v=\sigma/\mu$), the coefficient of
    asymmetry ($\gamma$) and the kurtosis ($\kappa$) of the eigenvalues of
    $\xcov$.
    \begin{equation}
      \label{eq:bound_expression}
      C_1 = D \frac{c_v^2}{1+c_v^2},
      \qquad \qquad
      C_2 = D \frac{c_v^4(\kappa - \gamma^2 -1) }
      {(\kappa -\gamma^2)c_v^4 + (\kappa-3-2\gamma)c_v^3 -2\gamma\, c_v +1}.
    \end{equation}
  \end{corollary}
  \begin{proof}
    These identities are obtained by expressing the
    raw moments that appear in $C_L$ in terms of $\mu$, $\sigma$, $\gamma$ and $\kappa$, and then simplifying the resulting formulas.
  \end{proof}
  From these expressions it is apparent that the more
  concentrated the eigenvalues of $\xcov$ ($c_v\rightarrow 0$), the fewer PLS
  components are needed to approximate $\olsEstim$ with a given accuracy.
  The bound for $L=1$ depends only on the coefficient of variation
  of the distribution of eigenvalues $c_v=\sigma/\mu$.  
  It is proportional to $c_v^2$ in the limit $c_v \rightarrow 0^+$.
  Therefore, if the eigenvalues of  $\xcov$ are grouped in one tight cluster, keeping a single PLS component yields an
  accurate approximation of $\hat{\boldsymbol\beta}_{\mathrm{OLS}}$. 
  The value of the bound for $L = 2$ depends not only on $c_v$ but also on the coefficient of asymmetry and the kurtosis, which makes it harder to interpret. 
  However, it is proportional to $c_v^4$ in the limit $c_v \rightarrow 0^+$.

  Additionally, from Pearson's inequality ($\kappa \ge 1 + \gamma^2$), the quantity $\kappa - \gamma^2 - 1$ is non-negative  
  \citep{sharmaSkewness2015}.  
  This quantity is zero for dichotomous distributions.
  Thus, $C_2$ should be small when the
  eigenvalues are distributed in two tightly packed clusters. 
  In the next section, we provide  numerical illustrations of the dependence of distance between $\plsEstim$ and $\olsEstim$ as a function of $L$, for different distributions of the eigenvalues of $\xcov$.

\section{Empirical study} 
  \label{sec:examples}
  In this section, an empirical study is carried out to investigate the effect of the eigenvalue distribution of the regressor covariance matrix on PLS.
  Specifically, we analyze the dependence of the quadratic-form distance between $\plsEstim$ and $\olsEstim$, the upper bound established in Theorem \ref{thm:bound} for this distance, and the accuracy of the linear predictor as a function of the number of PLS components considered.
  The analysis is first performed in regression problems with synthetic data for different forms of the distribution of eigenvalues. 
  The corresponding analysis is then performed for the California Housing dataset \citep{kelleypaceSparseSpatial1997}.

  \subsection{Synthetic data}

    In this section, synthetic data are used to illustrate the behavior of the PLS
    method depending on the eigenvalue distribution of the regressor covariance
    matrix. 
    Five regression problems are considered. 
    In these problems, $X$
    is modelled as a multivariate normal vector $X\sim N(0,\boldsymbol \Sigma)$.
    The eigenvalues of the covariance matrix $\boldsymbol \Sigma$, $\{\lambda_d\}_{d=1}^D$, are sampled from different distributions with specific characteristics. 
    Specifically, $D=30$ eigenvalues are selected with the following characteristics:
    \begin{enumerate}
      \setlength{\itemsep}{0pt}
      \item 30 equally spaced eigenvalues from 2.5 to 7.5.
      \item One cluster of 30 eigenvalues sampled from $N(5,0.1)$.
      \item Two clusters of 15 eigenvalues, each sampled from $N(2.5,0.1)$ and $N(7.5,0.1)$.
      \item Three clusters of 10 eigenvalues sampled from $N(2.5,0.1)$, $N(5,0.1)$, and
            $N(7.5,0.1)$.
      \item Three clusters of 10 eigenvalues sampled from $N(0.2,0.1)$, $N(5,0.1)$, and
            $N(7.5,0.1)$; so that one of the clusters is very close to zero.
    \end{enumerate}
    These eigenvalue distributions are displayed in Figure \ref{fig:eig_dist}.
    The actual covariace matrix is generated by a random rotation of the diagonal eigenvalue matrix: $\boldsymbol \Sigma = \mathbf Q^\top \mathrm{diag}
      (\lambda_1,\dots, \lambda_D)\mathbf Q$, where $\mathbf Q$ is a
    uniformly-distributed orthogonal random matrix. 
    The rotation matrix $\mathbf Q$ is obtained from the QR decomposition of a random matrix whose entries are sampled from a standard normal distribution \citep{mezzadriHowGenerateRandom2007}. 
    Finally, the data matrix, $\mathbf X = (\mathbf x_1,\dots, \mathbf x_N)^\top$  is
    obtained by stacking $N=1000$ samples from this random vector.

    To generate the response data, the linear model with additive noise presented
    in \eqref{eq:sample_reg} is used. 
    The $\boldsymbol \beta$ parameter is a random vector whose entries are sampled from a uniform distribution in $[0,1]$. 
    The noise $\boldsymbol\epsilon$ is sampled from a
    $N(0,\sigma^2)$ distribution, where $\sigma=0.1 \,\mathrm{std} (\mathbf X
      \boldsymbol \beta)$, so that the model is not dominated by the noise. 
    Finally, the response vector is computed as $\mathbf y=\mathbf X
      \boldsymbol \beta + \boldsymbol \epsilon$.
    
    In the experiments carried out, the closeness between $\plsEstim$ and $\olsEstim$ is quantified in terms of the normalized estimator difference: 
    \begin{equation}
      \label{eq:normal_diff}
      \mathrm{NED}_L =
      \frac{
      \|\plsEstim - \hat{\boldsymbol\beta}_{\mathrm{OLS}}\|_\xcov^2
      }
      {
      \|\hat{\boldsymbol\beta}_{\mathrm{OLS}}\|_\xcov^2
      },
    \end{equation}
    From \eqref{eq:final_bound}, it is apparent that $C_L$ is a bound on $\mathrm{NED}_L$.
    The results reported are averages over $20$ realizations of the data.

    The plots in the left column of Figure \ref{fig:pls_convergence}
    display the dependence of the normalized differences between the estimators, $\mathrm{NED}_L$, and of the corresponding upper-bound, $C_L$, on $L$, the number of PLS components considered. 
    These plots show how, as $L$ increases, the decrease of the bound introduced in Theorem \ref{thm:bound}  parallels that of the difference between the estimators. 
    As discussed in the previous section, PLS can be formulated as a polynomial fitting problem.
    In particular, Theorem \ref{thm:error_polynomial} provides a way of expressing the error of the estimator with $L$ iterations as a function of the values of some 
    polynomial $Q_L$, of degree lower or equal to $L$ that fulfills $Q_L(0)=-1$.     
    The optimal polynomials $Q^*_L$ defined in Theorem \ref{thm:error_polynomial} are plotted in the right column of Figure \ref{fig:pls_convergence} . 

    It is possible to interpret the features of the curves displayed in the
    left column of Figure \ref{fig:pls_convergence} from the characteristics of the polynomials plotted in the right column of this figure. 
    In the first scenario, in which the eigenvalues are uniformly distributed in an interval separate from zero, considering more components allows to find polynomials that fulfill $Q_L(0) = -1$ and take small values for all the eigenvalues.  
    In the second one, the decrease of $\mathrm{NED}_L$ is much steeper because having the eigenvalues closely packed in a single cluster makes the polynomial fitting problem much simpler. 
    For a given numbers of components, the corresponding polynomials take smaller values on the eigenvalues in the second scenario than in the first one.
    This result is consistent with the dependency of $C_1$ and $C_2$ on $c_v$ given in Corollary \ref{cor:explicit_bounds}.

    Figure \ref{fig:pls_convergence} also shows that the decrease of $\mathrm{NED}_L$ with $L$ follows different patterns depending on the number of clusters in which the eigenvalues are grouped. 
    In particular, the decrease is sharper for specific numbers of components. 
    When the eigenvalues are grouped in two clusters, the
    first abrupt decrease of $\mathrm{NED}_L$ occurs between $L=1$ and $L=2$ components. 
    This observation can be explained by noting it is not possible to find a polynomial of degree one (i.e., a straight line) that takes small values on both clusters and passes through the point $(0, -1)$. 
    However, a polynomial of degree two (i.e., a parabola) provides a reasonable fit.
    Significant improvements are observed also for $L=4$ and $L=6$ components.
    This is due to the fact that, in those cases, it is possible to find polynomials that pass through (0,-1) with equal numbers of roots located in the vicinity of each of the clusters.
    A similar analysis can be carried out for the fourth scenario, in which the eigenvalues are clustered in 3 groups. 
    In this case, sharper improvements are found for $L = 3$ and $L = 6$.

    To complete the analysis, we consider a case in which one of the clusters of eigenvalues is close to 0. 
    From the plot in the bottom left of Figure \ref{fig:pls_convergence} it is apparent that the decrease of $\mathrm{NED}_L$ with $L$ is rather slow. 
    The reason for this is that, since the fitted polynomial has to go through $(0, -1)$, large values of $L$ are needed so that the polynomial can take simultaneously small values for the eigenvalues in the vicinity of $0$ and in the other clusters.

    We now compare the performance of PCA and PLS regression as a function of $L$, the number of components considered. 
    The quality of the predictions is measured in terms of the coefficient of determination ($R^2$ score), which represents the proportion of explained variance.
    In most regression problems PLS is expected to outperform PCA because, in the definition of the components, the correlations between the regressor and response variables are taken into account in the former, but not in the latter \citep{frankChemometrics1993}. 
    Since the properties of PLS depend on the distribution of the eigenvalues of $\xcov$, the regressor covariance matrix, we carry out the analysis for the five scenarios described earlier.
    In Figure \ref{fig:r2_spread} we compare the curves that trace the dependence $R^2$ on $L$, for PLS (left plots) and PCA (right plots) in the first two synthetic datasets.    
    This comparison illustrates the differences between problems in which the eigenvalues of the regressor covariance matrix are uniformly distributed and problems in which they are clustered around a particular value, different from zero. 
    As expected, PLS obtains better results when the eigenvalues concentrate around a single value.
    In fact, when they are clustered in a single tight group, the PLS regression model with only one component provides a very accurate prediction of the response. 
    By contrast, when the eigenvalues are uniformly distributed, more components are needed.
    The behavior of PCA is markedly different. 
    In the case of clustered eigenvalues, the $R^2$ score of PCA increases linearly with the number of eigenvalues considered. 
    This is to be expected since the increment in explained variance is proportional to the eigenvalue that corresponds to the eigenvector considered by PCA at each step. 
    If the eigenvalues are spread out uniformly, PCA considers first the components that correspond to the larger eigenvalues. 
    Therefore, the magnitude of the eigenvalues decreases as more components are considered, which leads to a reduction of the rate at which $R^2$ increases for larger $L$.
    Additionally, Figure \ref{fig:r2_spread} also shows that PCA needs many more components to achieve the same $R^2$ scores that PLS.
    
    The plots displayed in Figure \ref{fig:r2_clusters} illustrate the properties of the curves that trace the evolution of the $R^2$ score as a function of $L$, depending on the number of clusters in which the eigenvalues are grouped.
    From these results we conclude that, in this case, the number of PLS components necessary to obtain a value of $R^2$ close to $1$ (perfect prediction) coincides with the number of eigenvalue clusters.
    This is consistent with the analysis of the differences between $\plsEstim$ and $\olsEstim$ for these datasets.
    Regarding  PCA, we can see how the number of clusters of eigenvalues has only minor effects in dependence of the $R^2$ scores with $L$. 
    For example, with two clusters, the $R^2$ increases faster during the first 15 iterations, which corresponds to the cluster with the largest 15 eigenvalues. For the scenario with three clusters of eigenvalues, the rate of increase drops after $10$ and $20$ components have been considered. 
    These correspond to having included in the model all the components in the first, and in the first and second largest clusters, respectively.

    Finally, we use the last two scenarios to investigate the impact of having a cluster of eigenvalues close to zero. 
    Figure \ref{fig:r2_zero} shows how that for $L > 1$, the performance of PLS deteriorates when there is a cluster of small eigenvalues. 
    This is again to be expected from the theoretical analysis carried out in the previous section because of the difficulties of fitting a polynomial that goes throgh $(0,1)$ and takes small values at the locations of the eigenvalues in the clusters.
    By contrast, PCA achieves better results when a sizeable fraction of the eigenvalues are close to zero.
    In fact, the maximum value of $R^2$ is attained for $L = 20$, once all the components that correspond to eigenvalues significantly larger than zero have been selected. 
    Nonetheless, PLS outperforms PCA regression also in these scenarios.
 
  \subsection{The Californian Housing dataset}

    In this section we analyze the properties of PLS regression for the California Housing dataset \cite{kelleypaceSparseSpatial1997}
    In this problem, the goal is to predict the median house value in a particular block group in a California district using 8 attributes ($D=8$): the median house age, the average number of rooms, the average number of bedrooms, the number of people residing within the block, the average number of household members, and the latitude and longitude of the block group.
    As a preprocessing step, both the regressor vector and the response variable are centered so that they have zero mean.
    Each column of $\mathbf X$ is scaled so that it has unit variance.
    In the original dataset, a median house value of $500,000 \$$ is assigned to instances whose actual value is above that threshold.
    To avoid distortions associated to this thresholding, these examples have been discarded.
 
    Figure \ref{fig:ch_dist} shows the eigenvalue distribution of the regressor covariance matrix for the California Housing dataset. 
    The eigenvalues are roughly grouped in three clusters, one of them close to zero. 
    This pattern is similar to the last synthetic dataset analyzed in the previous section.
    However, the eigenvalues in the central cluster are more spread out. 
    This dispersion hinders somewhat the performance of PLS, which is nonetheless fairly good. 
    The differences between the PLS and the OLS estimators as a function of $L$ are analyzed in Figure\ref{fig:ch_conv}.
    The left plot displays the dependence of these differences, quantified by $\mathrm{NED}_L$, and of $C_L$, the upper bound of these differences derived in this work, as a funtion of $L$, the number of PLS components considered. 
    Note that, for $L = 8$ the PLS coincides with the OLS estimator. 
    As expected, the distance between the estimators decreases slowly, because of the presence of the small eigenvalues and, to a lesser extent, the dispersion of the medium-sized eigenvalues.

    Figure \ref{fig:ch_r2} presents the results of a comparison between PCA and PLS regression.
    The left plot displays the curves that trace the dependence of the $R^2$ score, a measure of the quality of the predictions, with $L$, the number of components considered. 
    From these results one concludes that PLS obtains better results than PCA, and needs fewer components to achieve an accuracy comparable to OLS.
    The evolution of  $C_L$  as a function of $L$ is displayed in the right plot of this figure.
    Note that the descent of the bound mirrors the increase of the $R^2$ score as $L$ increases.
    This illustrates that the upper bound defined in Theorem \ref{thm:bound} provides an effective way to monitor the performance of PLS.

    \begin{figure}[p]
      \centering
      \includegraphics[width=0.5 \textwidth]{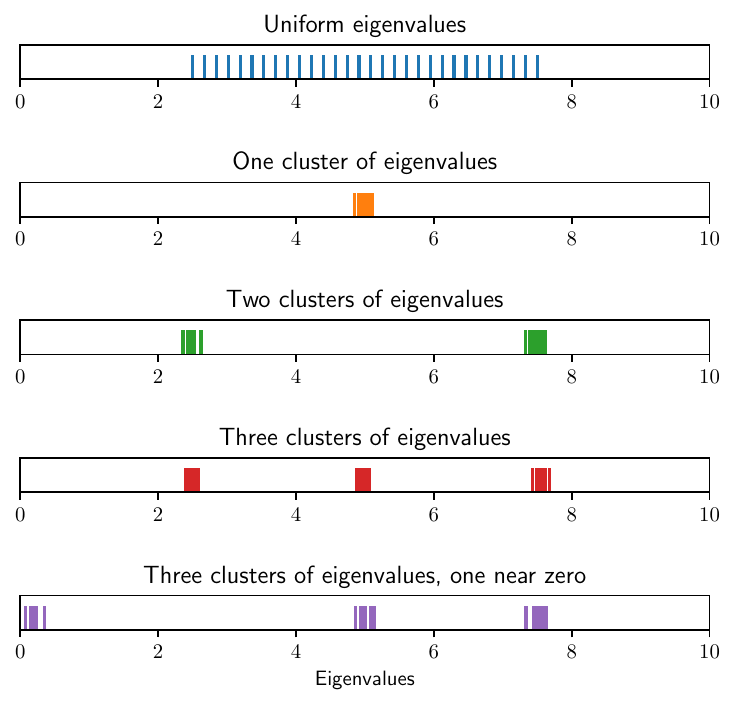}
      \caption{Eigenvalue distributions in the synthetic regression problems}
      \label{fig:eig_dist}
    \end{figure}

    \begin{figure}[ht]
      \centering
      \includegraphics[width=1 \textwidth]{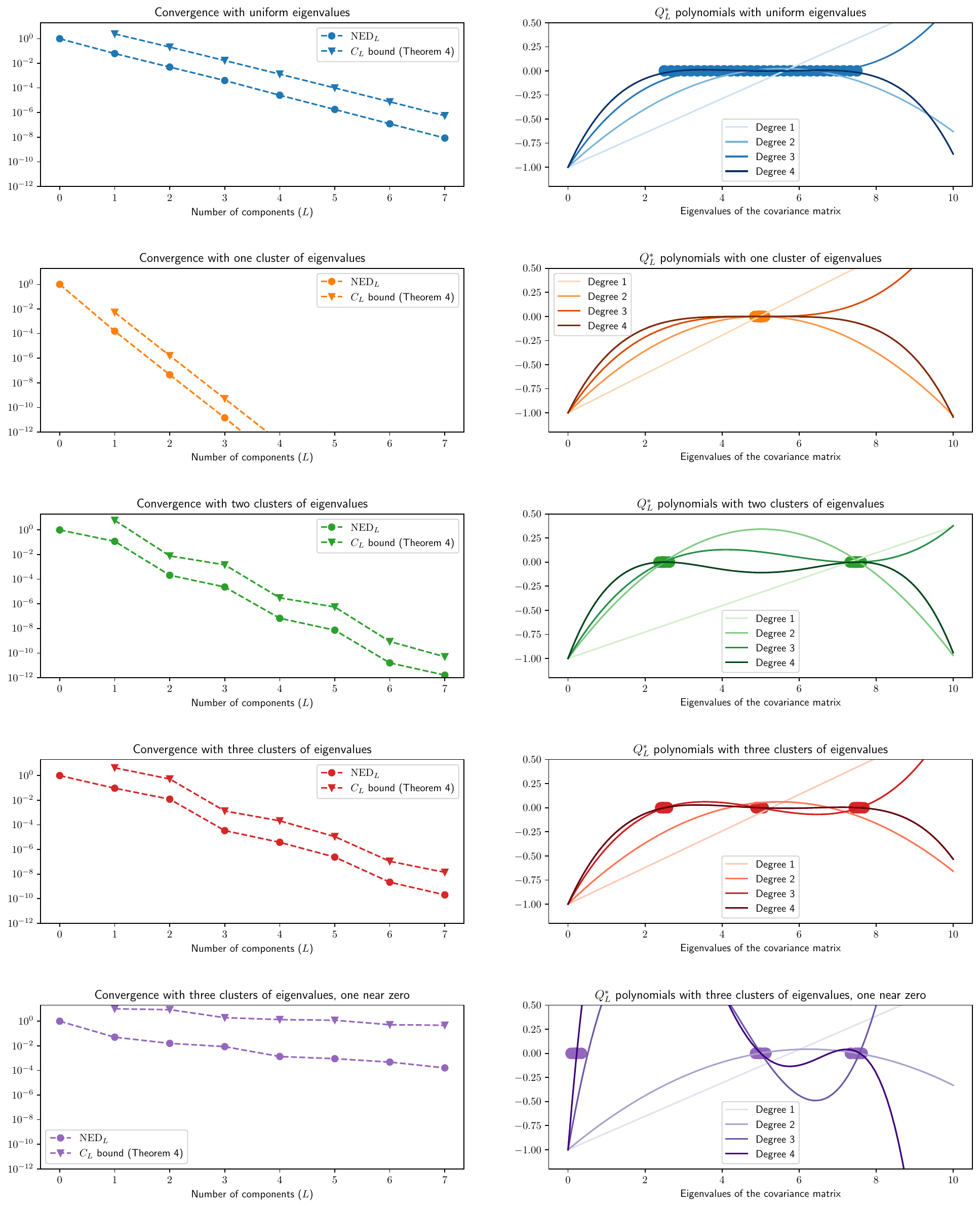}
      \caption{PLS estimator distance analysis with different distributions of eigenvalues
        of the regressor covariance matrix}
      \label{fig:pls_convergence}
    \end{figure}
    \begin{figure}[ht]
      \centering
      \includegraphics[width=1 \textwidth]{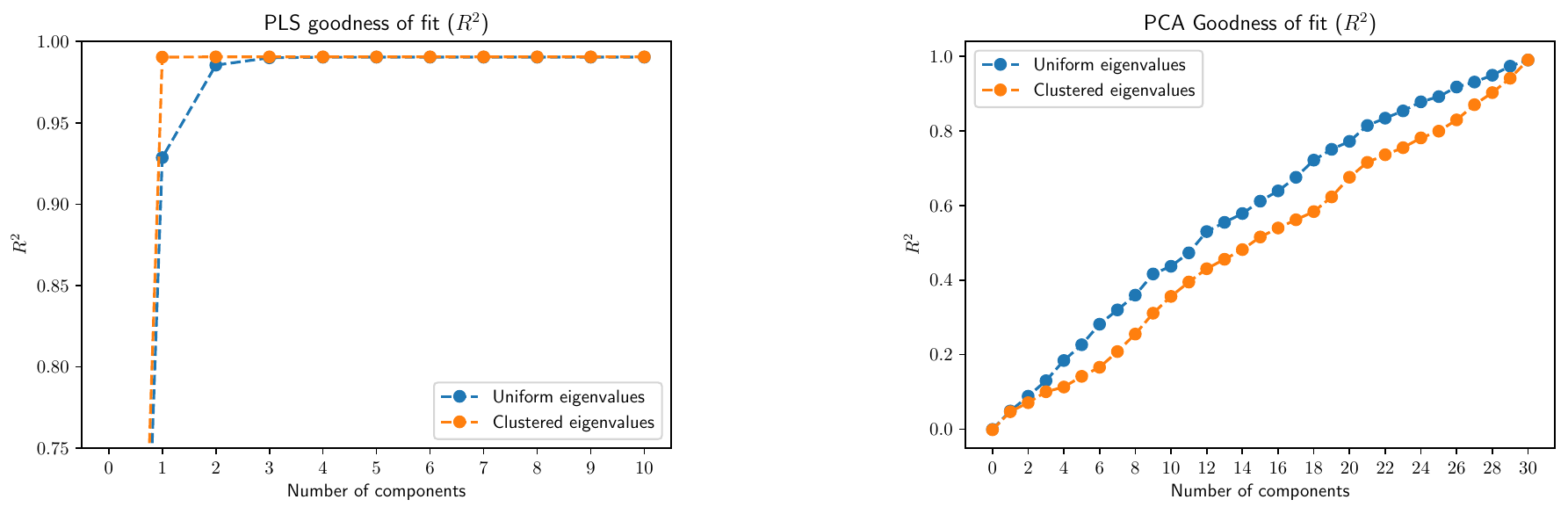}
      \caption{Accuracy of the predictions of PCA and PLS regression measured in terms of the $R^2$ score
      depending on whether the eigenvalues are concentrated or spread out uniformly.}
      \label{fig:r2_spread}
    \end{figure}
    \begin{figure}[ht]
      \centering
      \includegraphics[width=1 \textwidth]{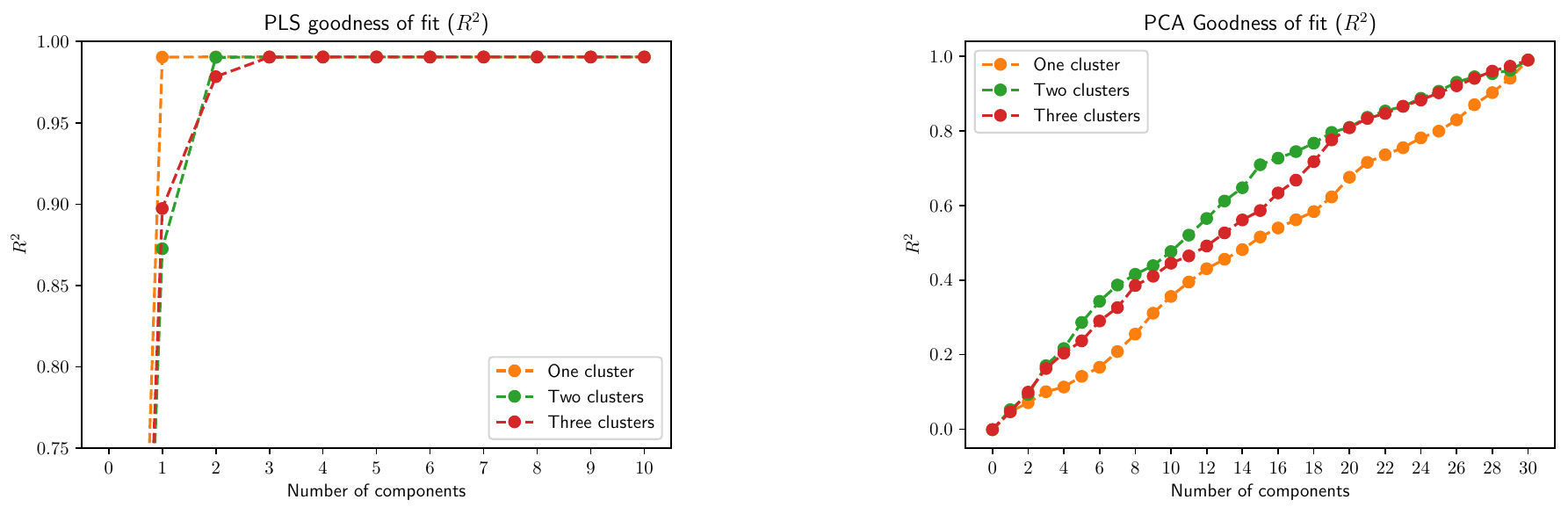}
      \caption{Accuracy of the predictions of PCA and PLS regression measured in terms of the $R^2$ score depending on the
        number of clusters in which the eigenvalues are grouped}
      \label{fig:r2_clusters}
    \end{figure}

    \begin{figure}[ht]
      \centering
      \includegraphics[width=\textwidth]{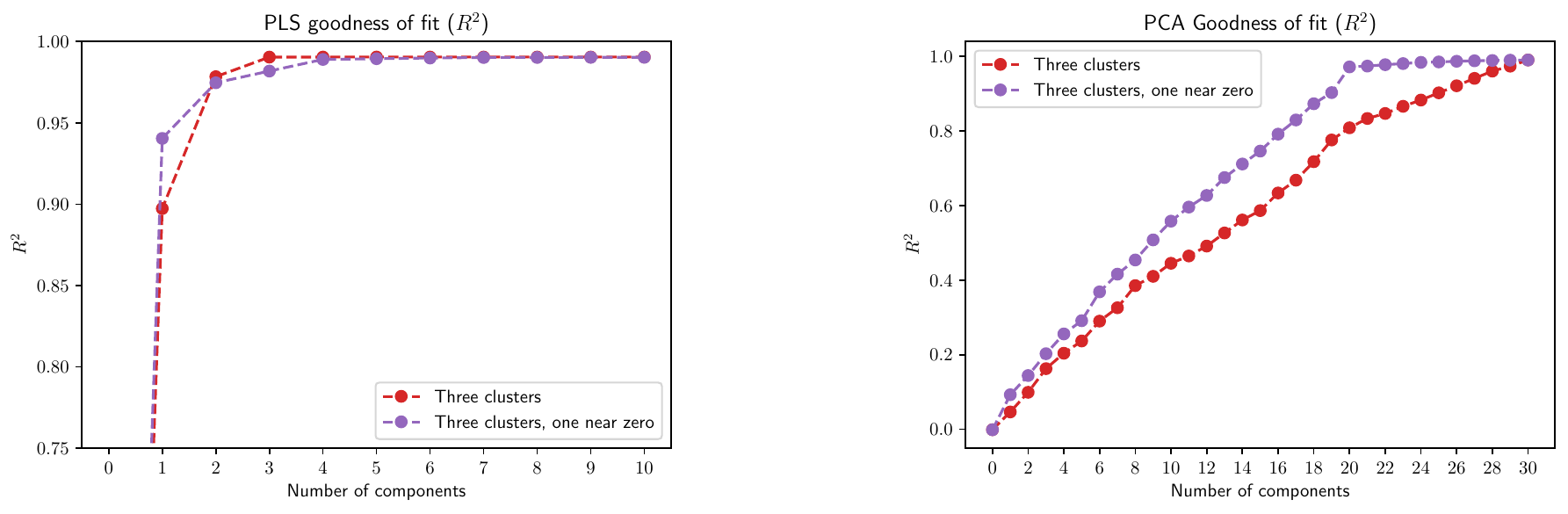}
      \caption{Accuracy of the predictions of PCA and PLS regression measured in terms of the $R^2$ score, depending on whether there is a cluster of eigenvalues near zero}
      \label{fig:r2_zero}
    \end{figure}

    \begin{figure}[ht]
      \centering
      \includegraphics[width=0.5 \textwidth]{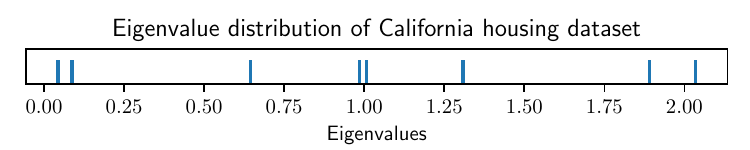}
      \caption{Eigenvalue distribution in the California Housing dataset}
      \label{fig:ch_dist}
    \end{figure}

    \begin{figure}[ht]
      \centering
      \includegraphics[width=0.9 \textwidth]{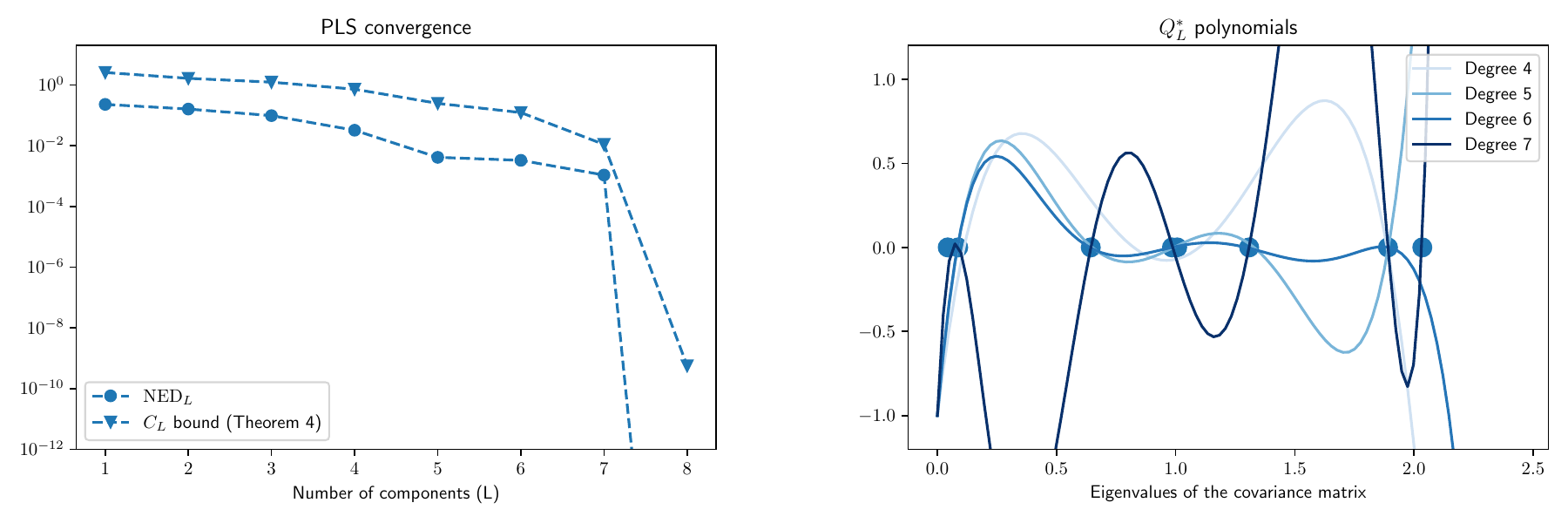}
      \caption{PLS estimator distance analysis in the California Housing dataset}
      \label{fig:ch_conv}
    \end{figure}

    \begin{figure}[ht]
      \centering
      \includegraphics[width=0.9 \textwidth]{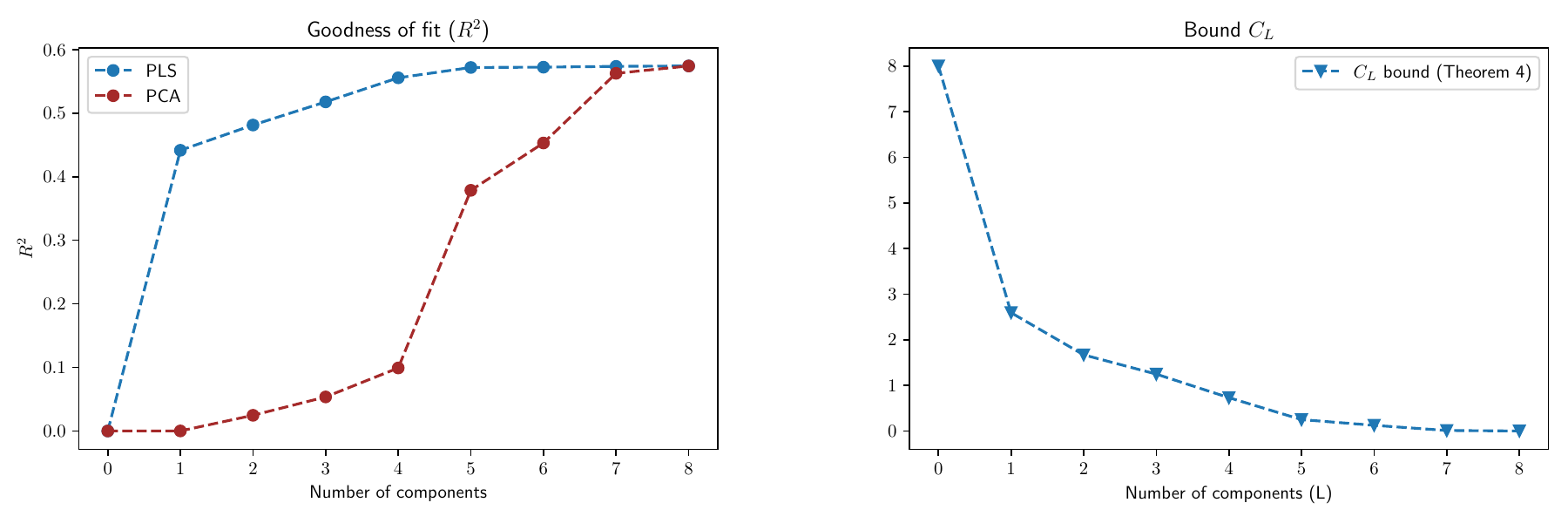}
      \caption{Accuracy of the predictions of PCA and PLS regression measured in terms of the $R^2$ score in the California Housing dataset}
      \label{fig:ch_r2}
    \end{figure}

\section{Conclusions}
  \label{sec:conclusions}
  
  In this work, the relation between ordinary least squares (OLS) and partial least squares (PLS) regression has been established by analyzing a number of different but equivalent optimization problems.
  In the context of scalar regression, the PLS components are orthogonal linear combinations of the regressor variables ($X$) that maximize the covariance with the response variable ($Y)$.
  A linear predictor is then built by taking a linear combination of a subset of size $L$  of the PLS components. 
  The coefficients in this linear combination are obtained by least squares.
  The PLS predictor can be expressed also as a linear combination of the original regression variables. 
  The estimate of the vector of regression coefficients in the original variables given by PLS with $L$ components, $\plsEstim$, is the solution of a restricted least squares problem in a Krylov subspace of order $L$ generated by $\mathrm{cov}(X, X)$, the  covariance matrix of X, and $\mathrm{cov}(X, Y)$, the vector of cross-covariances between X and Y.
  An important contribution of this work is to show that $\plsEstim$ is the vector of regression coefficients that is closest to the OLS estimator in this Krylov subspace. 
  Closeness is measured in terms of the Mahalanobis distance with the covariance matrix of the OLS estimator, which is a natural measure of differences in the space of estimators.
  Finally, leveraging the connection between optimization in Krylov subspaces and conjugate gradients, PLS regression is related to a polynomial optimization problem.
  From this reformulation, we derive an upper bound for the differences between the PLS and the OLS estimators of the vector of regression coefficients.
  This bound depends only on the eigenvalue distribution of the covariance matrix of the regressor variables. 
  In particular, PLS is expected to be most effective when the eigenvalues are not close to zero and appear tightly grouped in a few clusters.
  Furthermore, if the regressor covariance matrix has only $M$ distinct eigenvalues, convergence would be reached after $M$ steps.
  An empirical study using simulated data is carried out to analyze the effect of different types of eigenvalue distributions on the effectiveness of PLS.
  Finally, a similar studied is performed for the California Housing dataset.
  The results obtained illustrate the relevance of the theoretical analysis presented in this work and the advantages of PLS regression in real-world applications.  

\section*{Acknowledgments}
  J.R.B. acknowledges financial support from Grant CEX2019-000904-S funded by
  MCIN/AEI/ 10.13039/501100011033, and Spanish Ministry of Education and
  Innovation project PID2019-109387GB-I00.

  \noindent
  A.S. acknowledges financial support from project PID2022-139856NB-I00
  funded by MCIN/ AEI / 10.13039/501100011033 / FEDER, UE and project
  PID2019-106827GB-I00 / AEI / 10.13039/501100011033 and from the
  Autonomous Community of Madrid (ELLIS Unit Madrid).

  \appendix

\section{Proofs}
  \label{sec:proofs}

  \noindent\textbf{Proof of Proposition \ref{prop:nipals_properties}.}
  The identity for $\mathbf X$ in \eqref{eq:nipals:approx} is a direct consequence
  of substituting line 6 into line 7 of Algorithm \ref{alg:nipals:pls1}. The corresponding
  identity for $\mathbf y$ can be derived in a similar manner once one notices
  that adding a deflation step for $\mathbf y$ at the end of each iteration would
  not affect the results of NIPALS. This holds true since $\mathbf y$ is  used only 
  in the calculation of $\mathbf w_l$. After including the deflation step, 
  $\mathbf w_l = \mathbf X_{l-1}^\top \mathbf y_{l-1} / \|\mathbf X_{l-1}^\top \mathbf y_{l-1}\|$.
  However, $\mathbf X_{l-1}^\top \mathbf y_{l-1}=\mathbf X_{l-1}^\top \mathbf y$ since
  $\mathbf y_{l-1} = \mathbf y - \sum_{i=1}^{l-1}\mathbf t_i q_i$,  and
  $\mathbf X_{l-1}^\top \mathbf t_i = \mathbf 0$ as long as $i<l$.
  Regarding the decrease of the Frobenius norm, this is a consequence of the
  expressions for $\mathbf X_l$ and $\mathbf y_l$ in \eqref{eq:nipals:x_def}. We
  will prove the result for $\mathbf X_l$. From \eqref{eq:nipals:x_def}, one
  obtains
  $\mathbf X_l = \boldsymbol \Pi_l \mathbf X_{l-1}$, 
  where $\boldsymbol \Pi_l =\left(
      \mathbf I - \frac{\mathbf t_l \mathbf t_l^\top}
      {\mathbf t_l^\top \mathbf t_l}
      \right)
	  $.

  Then, to show the decrement of the norms, We need only show that $\|\mathbf
    X_l\|_F \le \|\mathbf X_{l-1}\|_F$ for $1\le l < L$:
  \begin{equation*}
    \|\mathbf X_l\|_F = \|\boldsymbol \Pi_l \mathbf X_{l-1}\|_F
    = \|\mathbf U_l \mathbf S_l \mathbf U_l^\top
    \mathbf X_{l-1} \|_F
    = \|\mathbf S_l \mathbf U_l^\top
    \mathbf X_{l-1} \|_F
    \le
    \|\mathbf U_l^\top
    \mathbf X_{l-1} \|_F
    =
    \|\mathbf X_{l-1} \|_F,
  \end{equation*}
  where $\boldsymbol \Pi_l=\mathbf U_l \mathbf S_l \mathbf U_l^\top$
  is the eigenvector decomposition of $\boldsymbol \Pi_l$. Since
  $\boldsymbol \Pi_l$ is a real symmetric matrix, $\mathbf U_l$ is
  a unitary matrix and we can apply that the Frobenius norm
  is invariant under unitary operations.
  Additionally, since $\boldsymbol \Pi_l$ is positive-definite
  and idempotent, its eigenvalues are either 0 or 1. Therefore,
  $\mathbf S_l$ has only 0s or 1s in the diagonal. As a result,
  multiplying by it can only reduce the Frobenius norm.

  The orthogonality between $\mathbf X_L$ and $\mathbf W$ can be proven showing
  that $\mathbf X_L \mathbf w_l=\textbf 0$ if $l\le L$. From
  \eqref{eq:nipals:x_def},
  \begin{equation*}
    \mathbf X_L \mathbf w_l =
    \left(
    \mathbf I
    -\frac{\mathbf t_L \mathbf t_L^\top}
    {\mathbf t_L^\top \mathbf t_L}
    \right)
    \dots
    \left(
    \mathbf I
    -\frac{\mathbf t_l \mathbf t_l^\top}
    {\mathbf t_l^\top \mathbf t_l}
    \right)
    \mathbf X_{l-1}
    \mathbf w_l
    =
    \left(
    \mathbf I
    -\frac{\mathbf t_L \mathbf t_L^\top}
    {\mathbf t_L^\top \mathbf t_L}
    \right)
    \dots
    \left(
    \mathbf I
    -\frac{\mathbf t_l \mathbf t_l^\top}
    {\mathbf t_l^\top \mathbf t_l}
    \right)
    \mathbf t_l
    = \textbf 0
  \end{equation*}

  Regarding the expressions for the loadings, both identities can be proven in
  the same way. We will prove the identity for $\mathbf P$, the $X$ loadings,
  showing the equality for each column of both sides of the equation. This
  equality is, in turn, a consequence of the expression for $\mathbf X$ in
  \eqref{eq:nipals:x_def}.
  \begin{equation*}
    \label{eq:nipals:loadings:dem}
    \begin{alignedat}{1}
      \mathbf X^\top \mathbf t_l \|\mathbf t_l\|^{-2}
       & =
      (\mathbf T_{l-1} \mathbf P_{l-1})^\top
      \mathbf t_l \|\mathbf t_l\|^{-2}
      +
      \mathbf X_{l-1}^\top
      \mathbf t_l \|\mathbf t_l\|^{-2}
      =
      \mathbf P_{l-1}^\top \mathbf T_{l-1}^\top
      \mathbf t_l \|\mathbf t_l\|^{-2}
      +
      \mathbf X_{l-1}^\top
      \mathbf t_l \|\mathbf t_l\|^{-2}
      = \mathbf p_l,
    \end{alignedat}
  \end{equation*}
  where $\mathbf T_{l-1}^\top \mathbf t_l = \mathbf 0$ because the extracted
  components are orthogonal.
  \hfill $\square$

  \noindent\textbf{Proof of Proposition \ref{prop:rotation}.}
  From Proposition \ref{prop:nipals_properties},
  $\mathbf X_L \mathbf W_L=\mathbf 0$. Applying this to the
  decomposition for $\mathbf X$ in \eqref{eq:nipals:approx}, we obtain:
  \begin{equation*}
    \pushQED{\qed}
    \label{eq:rotations:dem_1}
    \mathbf X \mathbf R_L
    =
    (\mathbf T_L \mathbf P_L^\top + \mathbf X_L)
    (\mathbf W_L (\mathbf P_L^\top \mathbf W_L)^{-1})
    =
    \mathbf T_L \mathbf P_L^\top \mathbf W_L (\mathbf P_L^\top \mathbf W_L)^{-1}
    +
    \mathbf X_L \mathbf W_L (\mathbf P_L^\top \mathbf W_L)^{-1}
    = \mathbf T_L.
    \qedhere
    \popQED
  \end{equation*}

  \noindent\textbf{Proof of Theorem \ref{thm:error_polynomial}.}
  Since $\xcov=\mathbf X^\top \mathbf X$ is a real, symmetric matrix, it is possible to
  find a sequence of non-negative eigenvalues $\{\lambda_1, \dots, \lambda_D\}$
  and orthonormal eigenvectors: $\{\mathbf u_1, \dots, \mathbf u_D\}$ such that
  $\xcov = \sum_{d=1}^D \lambda_d \mathbf u_d \mathbf u_d^\top$. This eigenvalue
  decomposition has three properties.
  First, the eigenvectors span the entire $\mathbb{R}^D$ space.
  Therefore, $D$ scalars $\{\xi_d\}_{d=1}^D$ can be found such that
  $\hat{\boldsymbol \beta}_{\mathrm{OLS}} = \sum_{d=1}^D \xi_d \mathbf u_d$.
  Second, the norm of a vector can be calculated as
  $\|\mathbf z\|_\xcov^2 = \mathbf z^\top \xcov \mathbf z
    = \sum_{d=1}^D \lambda_d (\mathbf u_d^\top \mathbf z)^2$
  Third, for any polynomial $P$, it holds that
  $P(\xcov)\mathbf u_d = P(\lambda_d) \mathbf u_d$, for
  $d=1,\dots, D$.
  Using these properties we can now find an expression to calculate
  $\Big\|\hat{\boldsymbol \beta}_L - \hat{\boldsymbol \beta}_{\mathrm{OLS}}\Big\|_\xcov^2$
  in terms of the polynomials $P_L^*$.
  \begin{equation}
    \label{eq:norm_estimate}
    \Big\|\hat{\boldsymbol \beta}_L - \hat{\boldsymbol \beta}_{\mathrm{OLS}}\Big\|_\xcov^2
    =
    \Big\|(P^*_{L-1}(\xcov)\xcov - \mathbf I)\hat{\boldsymbol\beta}_{\mathrm{OLS}}\Big\|_\xcov^2
    =
    \Big\|\sum_{d=1}^D(P^*_{L-1}(\xcov)\xcov - \mathbf I)\xi_d \mathbf u_d\Big\|_\xcov^2
    =
    \sum_{d=1}^D Q_L^*(\lambda_d)^2 \lambda_d \xi_d^2,
  \end{equation}
  where $Q_{L}^*(t)=tP_{L-1}^*(t)-1$, a polynomial of degree lower or equal to $L$ that fulfills  $Q_L^*(0)=-1$.
  Additionally, Corollary \ref{cor:polynomial_equiv} shows that $P^*_{L-1}$ is
  the polynomial that minimizes the RHS of \eqref{eq:norm_estimate} over all
  polynomials of degree lower or equal to $L-1$. 
  Therefore, $Q^*_{L}$ minimizes that same quantity
  over all the polynomials $Q_L$ of degree at most $L$ such that $Q_L(0)=-1$. That is to say, over $\Omega_L$. \hfill $\square$

  \noindent\textbf{Proof of Theorem \ref{thm:bound}.}
  All polynomials in $\Omega_L$ can be expressed as
  $R_L(t) = -1 + a_1 t + \dots + a_L t^L$ for some coefficients $a_1,\dots, a_L$.
  Therefore, as a function of the coefficients of the polynomials, the bound can be expressed as
  $h_L(a_1, \dots, a_L) =  \sum_{d=1}^D \left(-1+a_1 \lambda_d+ \cdots + a_L\lambda_d^L\right)^2$.
  To minimize this function, we calculate its gradient
  and determine the coefficients for which it is zero:
  \begin{equation}
    \label{eq:h_gradient}
    \frac{\partial h_L}{\partial a_l} = 2\sum_{d=1}^D(-1+a_1 \lambda_d + \dots + a_L \lambda_d^L)\lambda_d^l
    =
    -2\sum_{d=1}^D\lambda_d^l + 2a_1 \sum_{d=1}^D \lambda_d^{l+1} + \dots + 2a_L \sum_{d=1}^D \lambda_d^{l+L}
    =
    0,
    \qquad
    l=1,\dots, L.
  \end{equation}

  By rewritting these equations in terms of the sample raw moments of the eigenvalues,
  we obtain $
    a_1 \mu_{l+1}' + \dots + a_L \mu_{l+l}' = \mu_l',
	$
	for $l=1,\dots, L$.
  These equations can be expressed as the system $\mathbf H_L \mathbf a_L =
    \mathbf c_L$. Therefore, the coefficients that minimize $h_L$ are $\mathbf
    a_L^* = \mathbf H_L^{-1}\mathbf c_L$. Additionally, we express $h_L$ as
  \begin{equation}
    \label{eq:minimize_w_2}
    h_L(a_1, \dots, a_L) =
    (-1, \mathbf a_L)
    \begin{pmatrix}
      1           & \dots  & 1           \\
      \vdots      & \ddots & \vdots      \\
      \lambda_1^L & \dots  & \lambda_D^L
    \end{pmatrix}
    \begin{pmatrix}
      1      & \dots  & \lambda_1^L \\
      \vdots & \ddots & \vdots      \\
      1      & \dots  & \lambda_D^L
    \end{pmatrix}
    \begin{pmatrix}
      -1 \\
      \mathbf a_L
    \end{pmatrix}
    =
    (-1, \mathbf a_L)
    \begin{pmatrix}
      D            & D \mathbf c_L^\top \\
      D\mathbf c_L & D\mathbf H_L       \\
    \end{pmatrix}
    \begin{pmatrix}
      -1 \\
      \mathbf a_L
    \end{pmatrix}.
  \end{equation}

  Substituting the expression for $\mathbf a_L^*$ in the previous formula shows that
  $
    h_L(\mathbf a_L^*) = (-1, \mathbf H_L^{-1}\mathbf c_L)=
    D(1-\mathbf c_L^\top \mathbf H_L^{-1}\mathbf c_L  ).
	$
  Finally, note that the obtained coefficients $\mathbf a_L^*$ define the polynomial
  $R_L^*(t) = -1 + a_1^* t + \dots + a_L^* t^L$, which minimizes $H_2$.
  \hfill $\square$

  \nocite{*}
  \bibliography{PLSOLSRelation}%

\end{document}